\PassOptionsToPackage{hypertexnames=false}{hyperref}
\documentclass[aps,prl,reprint,twocolumn,amssymb,showpacs]{revtex4-1}
\pdfoutput=1

\input{preamble}

\begin{document}
\title{Free-fermion descriptions of parafermion chains and string-net models}

\author{Konstantinos Meichanetzidis}
\author{Christopher J. Turner}
\author{Ashk Farjami}
\author{Zlatko Papi\'c}
\author{Jiannis K. Pachos}
\affiliation{School of Physics and Astronomy, University of Leeds, Leeds LS2 9JT, United Kingdom}

\date{\today}

\begin{abstract}
Topological phases of matter remain a focus of interest due to their unique properties -- fractionalisation, ground state degeneracy,  and exotic excitations. While some of these properties can occur in systems of free fermions, their emergence is generally associated with interactions between particles. Here we quantify the role of interactions in general classes of topological states of matter in all spatial dimensions, including parafermion chains and string-net models. Using the interaction distance [Nat. Commun. {\bf 8}, 14926 (2017)], we measure the distinguishability of states of these models from those of free fermions. We find that certain topological states can be exactly described by free fermions, while others saturate the maximum possible interaction distance. Our work opens the door to understanding the complexity of topological models and to applying new types of fermionisation procedures to describe their low-energy physics.
\end{abstract}

\maketitle

{\sl Introduction.--} A striking feature of many-body systems is their ability to exhibit collective phenomena without analogue in their constituent particles.
Many recent  investigations into exotic statistical behaviours focus on topologically ordered systems~\cite{Wen1990} that support anyons~\cite{Leinaas1977,Wilczek1982}.
These systems, such as spin liquids~\cite{Anderson1987} and fractional quantum Hall states~\cite{Tsui1982}, exemplify the non-perturbative effects of interactions in many-electron systems. On the other hand, there are systems such as 2D topological superconductors, which support topological excitations -- Majorana zero modes~\cite{ReadGreen,Ivanov2001}.
These systems
can be modelled by free fermions
but lack topological order. 
Hence, a general question arises: is it possible (and with what accuracy) to  describe a given topological state, with anyonic quasiparticles, by a Gaussian state corresponding to some free system?  

In this paper we show that a broad class of topological states
admit free-fermion descriptions.
We use the interaction distance 
$D_\mathcal{F}$~\cite{Turner2017} to measure their distinguishability from free-fermion states. We consider parafermion chains~\cite{Fendley:2012hw},
which are symmetry-protected topological phases,
as well as  two- and three-dimensional string-nets~\cite{LevinWen1,WalkerWang}, and Kitaev's honeycomb lattice model~\cite{KitaevHoneycomb}. These models include RG fixed points of general families of topological systems with excitations that exhibit anyonic and parafermionic statistics. Generally, we find a broad distribution of $D_\mathcal{F}$ values for the ground states of these models. For example, we show that some models nearly maximise $D_\mathcal{F}$, while others have ground states that are Gaussian states with $D_\mathcal{F}=0$, even if their energy spectra cannot be given in terms of free fermions. We support these results by analytical arguments at the fixed points of various models. In addition, we provide strong numerical evidence that the same conclusions hold away from fixed points.
Thus, we establish $D_\mathcal{F}$ as a new measure of the complexity of topological models.
Moreover, as $D_{\cal F}$ is defined for a quantum \emph{state} instead of the full spectrum of a system (like the well-known Bethe ansatz techniques), our results open the way for investigating new types of fermionisation procedures for describing the low-energy physics of interacting systems with $D_\mathcal{F}=0$. 

{\sl Interaction distance.--} To quantify how far a given state is from any Gaussian state, we define the interaction distance~\cite{Turner2017} as $D_\mathcal{F}(\rho) = \min_{\sigma\in\mathcal{F}} D(\rho,\sigma)$, i.e., the minimal trace distance, $D(\rho,\sigma)$, between the reduced density matrix $\rho$ of a bipartitioned system and the manifold ${\cal F}$, which contains all free-fermion reduced density matrices, $\sigma$.
It was proven~\cite{Turner2017} that $D_{\cal F}$ can be expressed exclusively in terms of the entanglement spectrum~\cite{Haldane08} as
\begin{equation}
  \label{eq:DF}
  D_\mathcal{F}(\rho)= \frac{1}{2}\min_{\{\epsilon\}} \sum_a | \rho_a - \sigma_a (\epsilon)|,
\end{equation}
where $\rho_a$ and $\sigma_a $ are the eigenvalues of $\rho$ and $\sigma$, respectively, arranged in decreasing order~\cite{Turner2017}. For Gaussian states we have $\sigma_a = \exp(-\epsilon_0-\sum_j\epsilon_jn_j(a))$, where $\{\epsilon_j\}$ is the set of variational single-particle energies corresponding to free fermion modes and $n_j(a) \in \{0,1\}$ is the modes' occupation pattern corresponding to the given level $a$ of the entanglement spectrum~\cite{Haldane08,PeschelEisler}.
Intuitively, $D_\mathcal{F}$ is dominated by the low-lying part of the entanglement spectrum  and it reveals the correlations between the effective quasiparticles emerging from interactions~\cite{Haldane08}.
Hence, $D_\mathcal{F}$ is expected to be stable under perturbations that do not cause phase transitions~\cite{Turner2017}. We next apply this measure to quantify the distance of various topological states of matter from free fermion states.

{\sl Parafermion chains.--} The 1D Ising model can be mapped to the Majorana chain by means of a Jordan-Wigner transformation~\cite{Schultz,KitaevChain}.
Similarly, $\mathbb{Z}_{N>2}$ generalisations of the Ising model known as the clock Potts model can be expressed in terms of parafermions~\cite{FradkinKadanoff,Fendley}.
Parafermion zero modes may be physically realised at interfaces between 2D topological phases~\cite{Clarke:2013qda,Mong:2014ii}. They are described by the Hamiltonian
\begin{equation}
  \label{eqn:parafermions}
  H_{\mathbb{Z}_N} = - \sum_j \alpha^\dagger_{2j} \alpha_{2j+1} - f \sum_j \alpha^\dagger_{2j-1} \alpha_{2j} + \text{h.c.},
\end{equation}
where the parafermion operators satisfy the generalised commutation relations $\alpha_j \alpha_k = \omega \alpha_k \alpha_j $ for $k>j$, where $\omega = e^{i 2\pi/N}$ and $(\alpha_j)^N=1$. Majorana fermions correspond to $N=2$. Recently, phase diagrams of such models have been mapped out numerically~\cite{Motruk2013,Li2015}. Here we focus on the gapped regime away from the critical points or critical phases~\cite{Elitzur1979}, and neglect other possible terms (e.g., chiral phase factor) in the Hamiltonian. 

We first consider the system at its fixed point $f=0$ and we place the bipartition between regions $A$ and $B$ at a $(2j,2j+1)$-link, as shown in Fig.~\ref{fig:hills32} (Top). This gives a $N$-fold degenerate spectrum $\bar\rho(N)$, with $\bar\rho_a=1/N$ for all $a$~\cite{Fendley:2012hw}, where the overline, $\bar\rho$, denotes the density matrices with flat spectrum. We would like to determine the optimal free state corresponding to such a flat probability spectrum. 
Let $n$ be the greatest integer such that $2^n \le N$. We surmise that the optimal free fermion spectrum is of the form
\begin{equation}
  \label{eq:optimalguess}
  \sigma_\text{ansatz} \simeq \mathrm{diag}\big(N^{-1},\dots,N^{-1},p,\dots,p\big),
\end{equation}
where there are $2^n$ entries for each value $N^{-1}$ and $p$. Normalisation $\tr(\sigma_\text{ansatz})=1$ fixes $p = 2^{-n} - N^{-1}$.
This ansatz is an element of the variational class $\mathcal{F}$, hence $D(\bar\rho,\sigma_\text{ansatz})$ forms an upper bound for $D_\mathcal{F}(\bar\rho(N))$:
\begin{equation}
D_\mathcal{F}(\bar\rho{(N)}) \le 3 - \frac{N}{2^n} - \frac{2^{n+1}}{N}.
\label{eq:upperbound}
\end{equation}
To evaluate $D(\bar\rho,\sigma_\text{ansatz})$ we pad the spectrum of $\bar\rho(N)$ with zeros, a procedure always viable as it leaves the entropy invariant~\cite{Turner2017}.
We find that the numerically computed $D_\mathcal{F}(\bar\rho{(N)})$ is in remarkable agreement with this upper bound, as shown in Fig.~\ref{fig:hills32} (Bottom).
We analytically proved that the two values coincide for $N\leq 6$,
while we numerically verified it for up to $N = 2^8$~\footnote{Supplemental Online Material.}.
Hence, we conjecture that the upper bound of Eq.~\rf{eq:upperbound} is the \emph{exact} maximum of $D_\mathcal{F}(\bar\rho)$. This result  also applies to  the 2D and 3D models presented below.
\begin{figure}[t]
  \includegraphics[width =\linewidth]{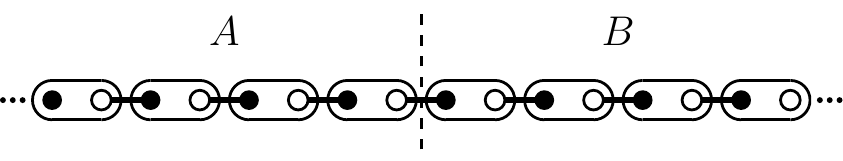}\\
  \includegraphics[width =\linewidth]{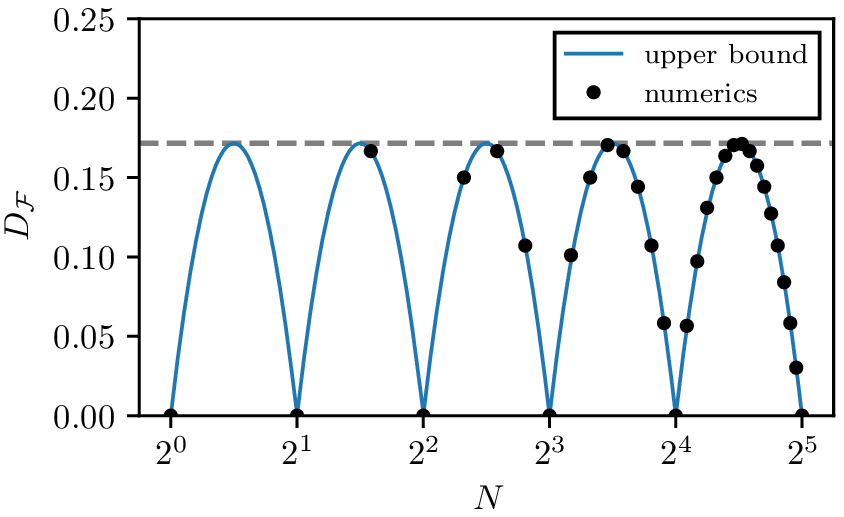}
  \vspace*{-0.8cm}
  \caption{(Top) Parafermion chain at a fixed point with bipartition into $A$ and $B$. (Bottom) Interaction distance, $D_\mathcal{F}(\bar\rho)$, for flat spectra of rank $N$. The solid blue line is the analytical upper bound given in \rf{eq:upperbound} and attains maximal value $D^\text{max}_\mathcal{F}$ (dashed line). The dots are results of the numerical optimisation and they coincide with the analytic upper bound.
  }
  \label{fig:hills32}
\end{figure}

From Eq.~\rf{eq:upperbound} we find that the maximum of the interaction distance is $D_\mathcal{F}^\text{max}=3 - 2\sqrt{2}$. This maximum is 
approached by rational approximations $N/2^n$ of $\sqrt{2}$ for increasing $n$, as shown in Fig.~\ref{fig:hills32} (Bottom).
By the exhaustive numerical maximisation $\max_{\rho} D_\mathcal{F}(\rho)$ for random $\rho$, we have not found states with interaction distance larger than $D_\mathcal{F}^\text{max}$~\footnotemark[\value{footnote}]. Hence, this appears to be the maximum possible value of the interaction distance \emph{for any state}.

The behaviour of the interaction distance for the flat spectra of parafermion chains, shown in Fig.~\ref{fig:hills32}~(Bottom), exhibits a recurring pattern, indicating that $\bar \rho$ has exactly the same interaction distance as $\frac{1}{2}(\bar\rho\oplus\bar\rho)$.
This doubling of the spectrum is equivalent to adding a zero fermionic mode to $\bar \rho$, which is
decoupled from the rest of the modes~\cite{Monogamy,MonogamyEdge},
and thus it is not expected to change its interaction distance. We conjecture that for a generic $\rho$, i.e., with a non-flat spectrum, we still have $ D_\mathcal{F}\left(\frac{1}{2}(\rho\oplus\rho)\right)=D_\mathcal{F}(\rho)$, which is supported by systematic numerical evidence~\footnotemark[\value{footnote}].

In conclusion, we find that $\mathbb{Z}_{N\neq 2^n}$ parafermion chains exhibit $D_\mathcal{F}\neq 0$, indicating that they are interacting in terms of complex fermions, while the inequality \rf{eq:upperbound} gave $D_{\cal F}=0$ for all $\mathbb{Z}_{2^n}$ models.
These results have been derived at the fixed point ($f=0$), and now we address their validity away from the fixed point. The entanglement spectrum and, as a consequence, the interaction distance, can distinguish between the universal and non-universal properties of gapped systems~\cite{Haldane08,Turner2017}.
When the parafermion chain is away from its fixed point,  it acquires a non-zero correlation length, $\xi$. To identify the universal properties of the system through $D_{\cal F}$, the linear size, $L_A$, of the partition $A$ should be $L_A \gg \xi$. The non-universal part is exponentially suppressed in a gapped phase and $D_\mathcal{F}$ predominantly describes the topological properties of the system, as shown in Fig.~\ref{fig:f=/0}. In this figure we see that $\mathbb{Z}_4$ has $D_{\cal F}=0$ for any value of $f$, while the interaction distance for $\mathbb{Z}_3$ approaches a step function through the phase transition. Hence, when the parafermion chain is away from criticality, its ground-state $D_\mathcal{F}$ is a robust characteristic of the topological phase. The value of $D_\mathcal{F}$ is accurately given by the upper bound, Eq. \rf{eq:upperbound}, for sufficiently large system and partition sizes.
\begin{figure}
  \includegraphics[width =\linewidth]{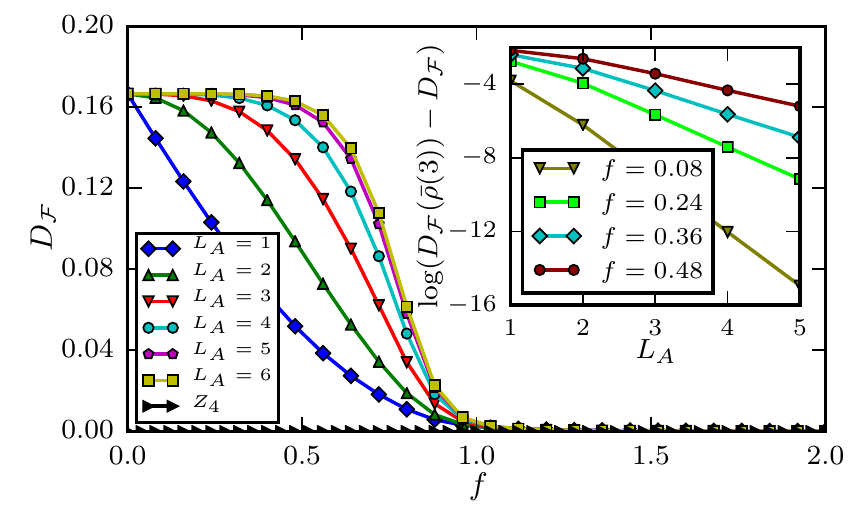}
  \caption{
  Interaction distance for the $\mathbb{Z}_4$ parafermion chains of length $L=8$ and $\mathbb{Z}_3$ of length $L=12$, with various partition sizes, $L_A$, as we move away from the fixed point, $f=0$. While $D_\mathcal{F}=0$ for $\mathbb{Z}_4$ for all values of $f$, it becomes a step function for $\mathbb{Z}_3$ as the partition size increases. (Inset) $\log(D_\mathcal{F}(\bar{\rho}(3))-D_\mathcal{F})$ for the $\mathbb{Z}_3$ chain shows that $D_\mathcal{F}$ converges exponentially to its fixed point value $D_\mathcal{F}(\bar{\rho}(3))$ as we increase $L_A$.}
  \label{fig:f=/0}
\end{figure}

We have also studied the excited states of parafermion models. In the $\mathbb{Z}_{2^n}$ cases, it can be shown that $D_\mathcal{F}=0$ for all excited states at the fixed point. However, at any finite $f>0$, the excited states in general have non-zero $D_\mathcal{F}$ (with the exception of single quasiparticle excitations above the ground state, which remain approximately free for $f>0$ close to the fixed point). This is consistent with the models, e.g., $\mathbb{Z}_{4}$, being  non-integrable for general $f$~\cite{Fendley:2012hw}, although its ground state remains Gaussian. 

To identify the free fermion description of the $\mathbb{Z}_{4}$ ground state we employ a matrix product state approach~\footnotemark[\value{footnote}]. For simplicity we restrict to the $f=0$ fixed point. One can express the ground state of $\mathbb{Z}_{4}$ in terms of the Gaussian ground state of $\mathbb{Z}_{2}\times\mathbb{Z}_{2}$ (which has the same entanglement spectrum) with the help of the  rotation, ${\cal U}= \otimes_j U_j$ ~\footnotemark[\value{footnote}], where the local unitaries $U_j$, acting on each site $j$ of the chain are given by
\be
U_j = 
\left( {\begin{array}{cccc}
		1 & 0 & 0 & 0 \\
		0 & 0 & 1 & 0 \\
		0 & v & 0 & v^* \\
		0 & v^* & 0 & v \\
	\end{array} } \right),
	\label{eqn:4to2U}
	\ee
where $v=\frac{1+i}{2}$. From this relation between the ground states we can identify the parent free fermion Hamiltonian, $H_{\mathbb{Z}_{4}}^\text{free} $, that has the same ground state as $H_{\mathbb{Z}_{4}}$
\be
H_{\mathbb{Z}_{4}}^\text{free} = {\cal U} H_{\mathbb{Z}_{2}\times \mathbb{Z}_{2}} {\cal U}^\dagger,
\label{eqn:MPS}
\ee
where $H_{\mathbb{Z}_{2}\times \mathbb{Z}_{2}}$ is defined similarly to the Hamiltonians in \rf{eqn:parafermions}.
This local Hamiltonian gives rise to the same zero modes localised at the end points of a chain like the  $\mathbb{Z}_{4}$ model. Nevertheless, their excitation spectra need not coincide. The construction in Eq.~(\ref{eqn:MPS}) can be extended to all $\mathbb{Z}_{2^n}$ models with $D_{\cal F}=0$. Note that in the case of $\mathbb{Z}_{2}$ parafermions the corresponding parent Hamiltonian, $H^\text{free}_{\mathbb{Z}_{2}}$, is identical to the $H_{\mathbb{Z}_{2}}$, which describes the Majorana chain.

{\sl String nets.--} We now turn to the string-net models~\cite{LevinWen1, LinLevin}. These are 2D RG fixed-point models that support topological order and anyon excitations. The models are defined in terms of
irreducible representations or `charges' of a finite group, $\mathcal{C}=\{1,...,n\}$, that parametrise the edges of a honeycomb lattice,
as shown in Fig.~\ref{fig:2D} (Left).
These charges obey the fusion rules $x\times y=\sum_{z}N_{xy}^zz$, where $N_{xy}^z$ is the multiplicity of each fusion outcome.
For each charge, $x$, the quantum dimension $d_x$ is defined that satisfies $d_x\times d_y=\sum_{z}N_{xy}^zd_z$.

The ground state of a string-net model can be interpreted as a superposition of all configurations of charge loops. The probability spectrum from any bipartition into single component regions $A$ and $B$  is determined by all string configurations, $a$, on the boundary $\partial A$ which fuse to the vacuum
\begin{equation}
  \rho_a=\frac{\prod_{j\in a}d_{x_j}}{\mathcal{D}^{2(|\partial{A}|-1)}},
  \label{eqn:stringnet}
\end{equation}
where $\mathcal{D}=\sqrt{\sum_x d_x^2}$ is the total quantum dimension of the group and $x_j$ is an element of the configuration, $a$, of charges at the boundary links, as shown in Fig.~\ref{fig:2D}~(Left). Hence, we can directly evaluate the entanglement spectra of all string-net models, Abelian or non-Abelian, and for any partition~\cite{Alex}.
  
We initially consider string-nets defined with an Abelian group $\mathbb{Z}_N$. These models have $d_x=1$ for all $x\in \mathcal{C}$ and thus $\mathcal{D}=\sqrt{N}$. From Eq.~\rf{eqn:stringnet} we find that the corresponding probability spectrum for any bipartition is flat with degeneracy $N^{|\partial{A}|-1}$. Hence, the interaction distance $D_\mathcal{F}$ is directly determined from Eq.~\rf{eq:upperbound} as in the case of parafermion chains.

The cases with $N=2^n$ can be exactly described by fermionic zero modes, giving $D_{\mathcal{F}}=0$ for any partition size. Hence, the ground states of these models are Gaussian states. This is a surprising result as anyonic quasiparticles are expected to emerge in interacting systems. Nevertheless, the optimal free states are
not necessarily local and their energy spectrum is not necessarily given by filling of single fermion modes.
For $N=2$ we obtain the well known Toric Code~\cite{KitaevToricCode}. We now show that the fermionisation of this model is given in terms of free lattice fermions coupled to a $\mathbb{Z}_2$ gauge field.

Kitaev's honeycomb lattice model~\cite{KitaevHoneycomb} is an interacting model that supports vortices with Abelian Toric Code or non-Abelian Ising anyonic statistics, depending on its coupling regime. Nevertheless, for fixed vortex configurations its Hamiltonian is reduced to free fermions living on the vertices of the honeycomb lattice coupled to a static $\mathbb{Z}_2$ gauge field $u$ that resides on its links~\cite{KitaevHoneycomb}. When we bipartition the ground state of the system the reduced density matrix splits into a gauge and a fermionic part, i.e. $\rho=\bar\rho_u\otimes \rho_\phi$~\cite{YaoQi}. The gauge part corresponds to a $\mathbb{Z}_2$ flat spectrum giving $D_\mathcal{F}(\bar\rho_u)=0$ and the fermionic part corresponds to free fermions with $D_\mathcal{F}(\rho_\phi)=0$. 
This means that $\rho$, as a tensor product of free fermion entanglement spectra, has also $D_\mathcal{F}(\rho)=0$ for any partition, rendering the ground state Gaussian. Hence, free fermions coupled to a $\mathbb{Z}_2$ gauge field provide the fermionisation prescription of the $\mathbb{Z}_2$ string-net model.

\begin{figure}
	\includegraphics[width =\linewidth]{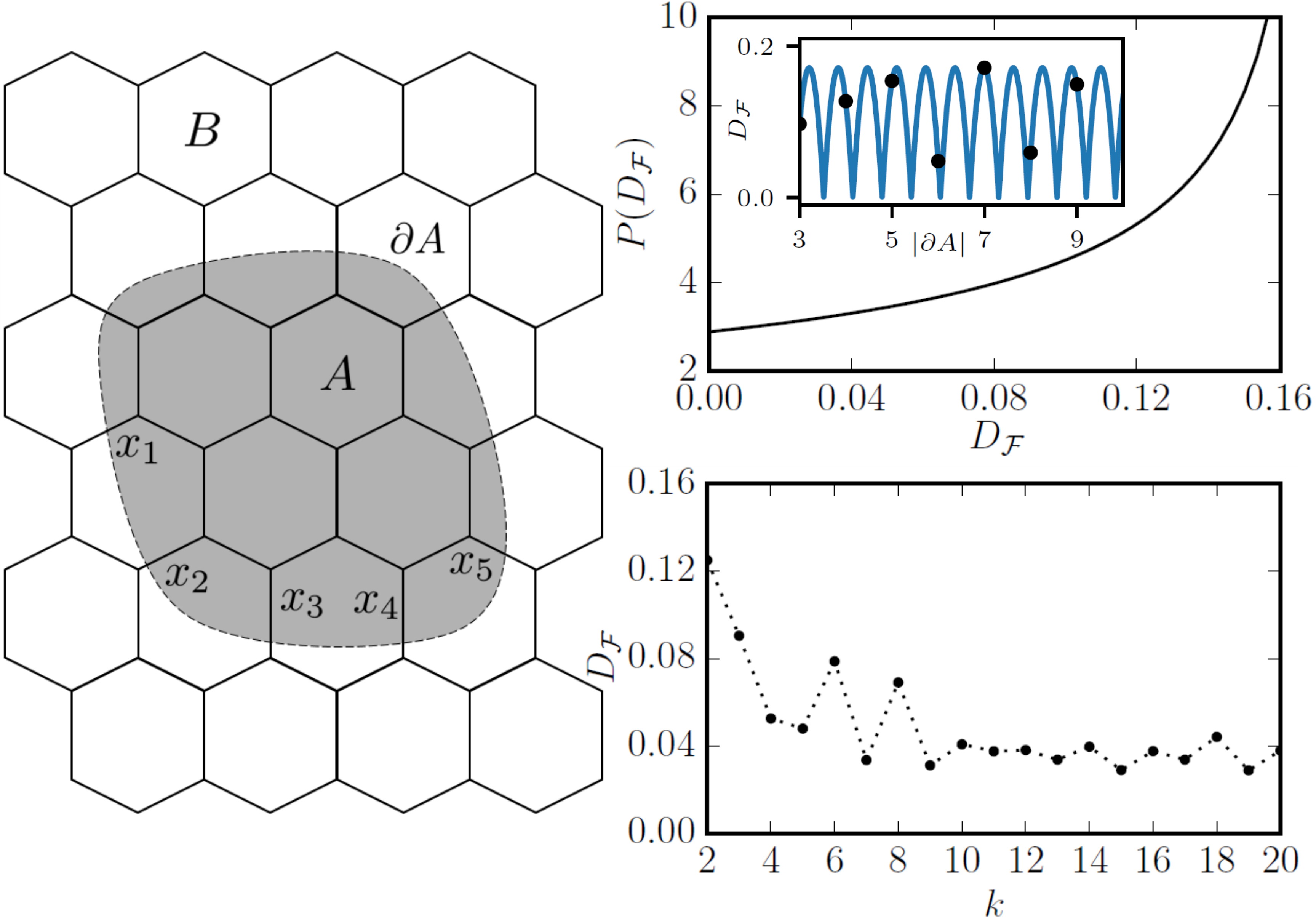} 
	\caption{(Left) The string-net model on a honeycomb lattice with a bipartition into $A$ and $B$. A configuration of charges $x_j$ is depicted at the links of the boundary $\partial A$. (Right, Top) Distribution $P(D_\mathcal{F})$ of the interaction distance for varying $|\partial A|$. (Inset) The dots represent the numerically obtained interaction distance as a function of $|\partial A|$. (Right, Bottom). Plot of $D_\mathcal{F}$ for SU(2)$_k$ against $k$ for a partition with $|\partial A|=3$.
	}
	\label{fig:2D}
\end{figure}

For string nets with $N\neq 2^n$, $D_\mathcal{F}$ is always non-zero. In particular, its value depends on the size $|\partial{A}|$ of the partition boundary. We investigate its behaviour by studying the distribution $P(D_\mathcal{F})$ of $D_\mathcal{F}$ by varying the size $|\partial{A}|$ of the boundary for a certain model $\mathbb{Z}_N$. This distribution can be shown to be given by
$ P(D_{\mathcal{F}})=\frac{2}{\ln 2}/\sqrt{1+D_{\mathcal{F}}(D_{\mathcal{F}}-6)}$,
which, surprisingly, is independent of $N$~\footnotemark[\value{footnote}]. Hence, there exist partitions that asymptotically maximise $D_{\cal F}$  for all $N\neq 2^n$, as shown in Fig.~\ref{fig:2D}~(Right, Top). Therefore, all $\mathbb{Z}_N$ Abelian string-nets either admit a free-fermion description for any partition or they form a class for which
the manifestations of interactions are equivalent.

We next consider the non-Abelian string-net models. For concreteness, we take the finite group to be $SU(2)_k$ for various levels $k\geq2$. This group gives rise to string-net models that support a large class of non-Abelian anyons, such as the Ising anyons for $k=2$, with statistics similar to Majorana fermions, or the Fibonacci anyons for $k=3$, that are universal for quantum computation~\cite{Fibonacci,JiannisBook}. For simplicity we consider the interaction distance for a single site partition that has $|\partial A|=3$. We find that $D_{\cal F}\neq 0$ for all $k\leq20$, as shown in Fig.~\ref{fig:2D} (Right, Bottom). Hence, it is not possible to find a free fermion description of these non-Abelian string-net models. Nevertheless, it is possible to have chiral non-Abelian models that are not RG fixed points, which admit a description of their ground state in terms of free fermions. As we have seen, Kitaev's honeycomb lattice model falls in this category.

The string-net construction given above for 2D topological models directly generalises to 3D topological systems, with entanglement spectra also given by Eq.~\rf{eqn:stringnet}. A more powerful generalisation is in terms of the Walker-Wang models, which have a rich behaviour in their bulk and at their boundary~\cite{WalkerWang}. Investigation similar to the  string-net SU(2)$_k$ models shows that the interaction distance depends not only on the size of $|\partial A|$, but also on the topology of $A$
due to braiding~\footnotemark[\value{footnote}].

{\sl Conclusions.--} We have quantified the effect of interactions in the ground states of broad classes of topological phases of matter in all spatial dimensions. For parafermion chains at the fixed point, 
we found that the partition size does not affect the value of $D_{\cal F}$. In contrast, for 2D string nets the size of the boundary matters, but not its geometry, while for Walker-Wang 3D models the topology of the boundary also becomes relevant.

Surprisingly, we discovered that the $\mathbb{Z}_{2^n}$ parafermion chains, as well as $\mathbb{Z}_{2^n}$ string-nets and Walker-Wang models, all have ground states with $D_{\cal F} = 0$ for any bipartition at the fixed point. Based on similar arguments, it is possible to show in such cases that $D_{\cal F}=0$ holds also for their excited states. Hence, the exciting possibility arises that the fermionisation procedure we applied to the $\mathbb{Z}_4$ parafermion model could be extended to all these states.
Moreover, we numerically demonstrated that $D_{\cal F}\approx 0$ continues to hold in the ground state and low-lying excitations, even when the system is away from the fixed point. However, the highly excited states typically have $D_{\cal F}\neq 0$ away from the fixed point.

Identifying Gaussianity in the low-lying energy eigenstates of a model can refine the notion of fermionisation procedures employed to solve quantum Hamiltonians. Such procedures are applicable, in the usual sense, if a system has $D_{\cal F}=0$ for all possible bipartitions and in all its eigenstates, while at the same time the energy spectrum is also free. A more subtle possibility appears when $D_{\cal F} = 0$ for some eigenstates (and all cuts), but the energy spectrum is not that of free fermions. We believe integrable systems~\cite{Sutherland} fall into this category. We also note that $D_{\cal F}=0$ in low-lying eigenstates is in principle compatible with anyon statistics because the latter only emerges when one interpolates adiabatically between different sectors of the conserved charges of the model~\cite{VilleNJP}. Finally, in some non-integrable cases like the $\mathbb{Z}_4$ parafermion model away from the fixed point, we found that $D_{\cal F}$ can surprisingly be zero in the ground state and the low-lying excited states of the system. This opens up the exciting possibility of describing the low-energy sectors of such interacting systems in terms of new types of free-particle models. 

{\bf Acknowledgements.} We thank Paul Fendley, Alex Bullivant and Jake Southall for inspiring comments. 
This work was supported by the EPSRC grants EP/I038683/1, EP/M50807X/1 and EP/P009409/1. Statement of compliance with EPSRC policy framework on research data: This publication is theoretical work that does not require supporting research data.
  
\bibliography{references}
\bibliographystyle{apsrev4-1}

\newpage
\onecolumngrid

\begin{center}
\textbf{\large Supplemental Online Material for ``Free-fermion descriptions of parafermion chains and string-net models'' }\\[5pt]
\begin{quote}
{\small In this Supplementary Material we discuss the ansatz and the derived upper bound for flat spectra. We then present numerical evidence that the upper bound is the exact solution for flat spectra. We provide numerical evidence for the zero-mode conjecture $ D_\mathcal{F}\left(\frac{1}{2}(\rho\oplus\rho)\right)=D_\mathcal{F}(\rho)$ and the conjecture that the maximum value of the interaction distance is $D_\mathcal{F}^\text{max}=3-2\sqrt{2}$.
Furthermore, we support in detail the statement made in the main text on the existence of an exact mapping of the ground state of the $\mathbb{Z}_4$ chain to that of $\mathbb{Z}_2 \times \mathbb{Z}_2$ using the formalism of matrix product states.
Following on, we offer a simple proof for the exactness of the upper bound of $D_\mathcal{F}$ for the rank-$6$ maximally entangled state, inspired by ideas from order theory applied on free-fermion spectra. Next, we calculate the distribution functions of $D_\mathcal{F}$ which characterise the parafermion models and string-nets. Finally, we calculate $D_{\cal F}$ for the topologically distinct subregions of a 3D Walker-Wang model.}\\[20pt]
\end{quote}
\end{center}
\twocolumngrid

\setcounter{equation}{0}
\setcounter{page}{1}
\setcounter{figure}{0}

\appendix

{\sl Optimal free ansatz for flat spectra.--} Here we describe in detail the construction of the optimal free-fermion description of a
flat $N$-rank entanglement spectrum.
We guess that it is optimal to exactly match as many of the highest probability levels as possible.
The most that can be matched are $2^n$, where $n$ is the greatest integer such that $2^n \le N$.
Considering the entanglement energies,
i.e. negative logarithms of the eigenvalues of $\rho$,
this is achieved by requiring $\sigma$ to contain $n$ zero-modes.
We further assume that the optimal $\sigma$ has only one further non-trivial mode whose splitting parameter is now fixed by the requirement that the first $2^n$ levels of $\sigma$ have eigenvalue $N^{-1}$ leading to the form given in the main text:
\begin{equation}
  \sigma \simeq \mathrm{diag}\big(N^{-1},\dots,N^{-1},p,\dots,p\big),
\end{equation}
Then we can write $\sigma$
as a tensor product of two-level (fermion) modes as
\begin{equation}\label{eq:appoptimalguess}
  \sigma \simeq \bigotimes_{i} \mathrm{diag}\left(\frac{1}{2}+s_i,\,\frac{1}{2}-s_i\right),
\end{equation}
where for $n$ of the modes $s_i=0$ and the one remaining mode $s_i = N^{-1}2^n - 1 / 2$,
ensuring that $\sigma$ is free.

There are two contributions to the upper-bound $D(\bar\rho,\sigma)$.
The first is from the entanglement levels with index $2^n + 1 \le k \le N$ for which the probability difference is between $N^{-1}$ and $p$,
illustrated by the dashed lines between the top and middle rows of eigenvalues in Fig.~\ref{fig:Probabilities}.
The second is from levels with index $N + 1 \le k \le 2^{n+1}$ for which the probability difference is between $0$ and $p$,
illustrated by the dashed lines between the middle and bottom rows of eigenvalues in Fig.~\ref{fig:Probabilities}.
In the trivial case where $N=2^n$ for some $n$, then $\bar\rho(N)$ can be reproduced by $n$-many zero-modes and the spectrum is free with $D_\mathcal{F}=0$ as seen in Fig.~\ref{fig:hills_on_fire}.

\begin{figure}[!htb]
  \includegraphics[width=0.95\columnwidth]{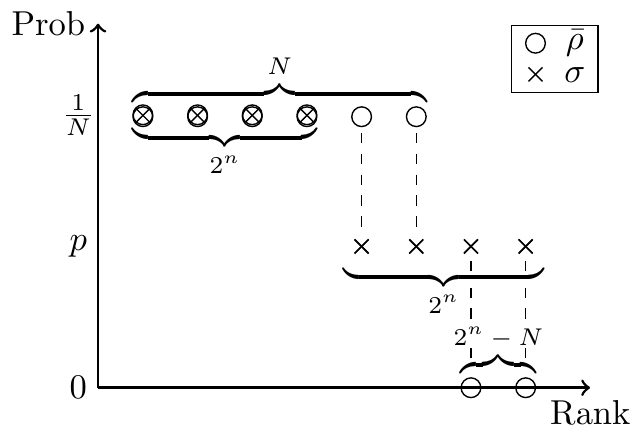}
  \caption{
    Probability spectra showing the flat spectrum of degeneracy $N$ (including zeros padding the spectrum to the next power of two) and the corresponding $\sigma_\text{ansatz}$ defined in Eq.\rf{eq:appoptimalguess}.
  }
  \label{fig:Probabilities}
\end{figure}

We compared this analytic upper bound with results from numerical optimisation
 for $N$ up to $2^8$, where we use the same Monte-Carlo basin hopping strategy used successfully previously for the 1D quantum Ising model in a magnetic field~\cite{Turner2017}.
For larger $N$ the increasing size of the optimisation problem results in the numerical minimisation failing to consistently find the global minimum.
Remarkably, numerical minimisation never finds results below the analytic upper-bound making a convincing case that this upper bound is in fact the exact result.

\begin{figure}[!htb]
  \includegraphics{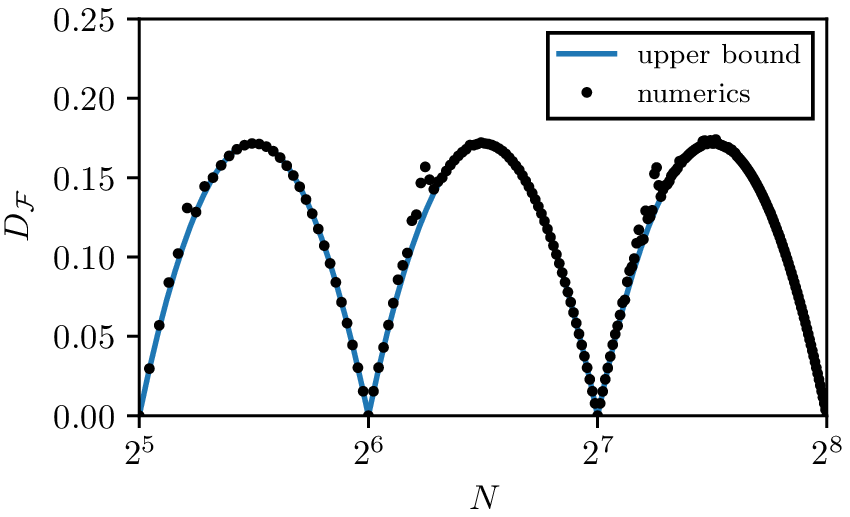}
  \caption{
    Interaction distance $D_\mathcal{F}(\bar\rho)$ for flat spectra of rank $N$.
    Results of numerical optimisation and the analytic upper bound $D_\mathcal{F}(\bar\rho{(N)}) \le 3 - \frac{N}{2^n} - \frac{2^{n+1}}{N}$ are compared, revealing that numerical optimisation never improves on the analytic upper-bound.
    The numerical data features deviations above the analytic curve which represent intermittent failure of finding the global minimum.
  }
  \label{fig:hills_on_fire}
\end{figure}

{\sl Numerical evidence for conjectures.--} Here we provide numerical evidence for conjectures regarding $D_\mathcal{F}$ which we state in the main text.
In particular, we indicate that the zero-mode conjecture, $D_\mathcal{F}(\rho)=D_\mathcal{F}(\frac{1}{2}(\rho\oplus\rho))$, holds,
as well as that the maximal possible value of the interaction distance is $D^\text{max}_\mathcal{F}$.
As stated in Ref.~\cite{Turner2017}, we use a basin-hopping algorithm which
collects values of local minima via the Nelder-Mead method for each basin of the cost function landscape.
We show that apparent violations of these conjectures
correspond to cases where the optimisation does not reach a global minimum
and increasing the number of basins results in the conjectures holding true.

Let $\rho$ represent a generic diagonal density matrix.
Its doubly degenerate version, $\frac{1}{2}(\rho\oplus\rho)$,
can be interpreted as the result of adding a zero-entanglement-energy in the system,
where entanglement energy is defined as the negative logarithm of a probability~\cite{Haldane08}.
This becomes clear when considering that each member of a degenerate pair of eigenvalues
corresponds to the zero-mode being occupied or not.
We compute $D_\mathcal{F}^\text{diff} =  D_\mathcal{F}(\rho)-D_\mathcal{F}(\frac{1}{2}(\rho\oplus\rho))$
for random diagonal density matrices.
The distribution $P(D_\mathcal{F}^\text{diff})$ is peaked at zero indicating the validity of the zero-mode conjecture,
as shown in Fig.\ref{fig:Conjectures} (Top).
The peak is more prominent when the number of basins is increased.
Note that $D_\mathcal{F}(\frac{1}{2}(\rho\oplus\rho)) \leq D(\frac{1}{2}(\rho\oplus\rho),\frac{1}{2}(\sigma\oplus\sigma)) = D_\mathcal{F}(\rho)$, where $\sigma$ is the optimal free state of $\rho$.
The inequality holds because $\frac{1}{2}(\sigma\oplus\sigma)\in\mathcal{F}$ is a member of the variational class $\mathcal{F}$.
Thus, we attribute $D_\mathcal{F}^\text{diff}<0$ to failure of the minimisation in finding the global minimum for 
$D(\frac{1}{2}(\rho\oplus\rho)$ due to the greater number of input probabilities.

Regarding the maximal possible interaction distance,
we use yet again basin-hopping
in order to perform the maximisation
$\max_{\rho} D_\mathcal{F}(\rho)$.
For each instance of $\rho$, we find $D_\mathcal{F}$ by the Nelder-Mead method,
with the condition that if the cost function is greater than the conjectured value $3-2\sqrt{2}$,
then basin-hopping is performed for a large enough number of basins so that finding the global minimum
is better approximated, which we never find to be above $3-2\sqrt{2}$ as shown in Fig.\ref{fig:Conjectures} (Bottom).

\begin{figure}[!htb]
  \includegraphics[width=0.95\columnwidth]{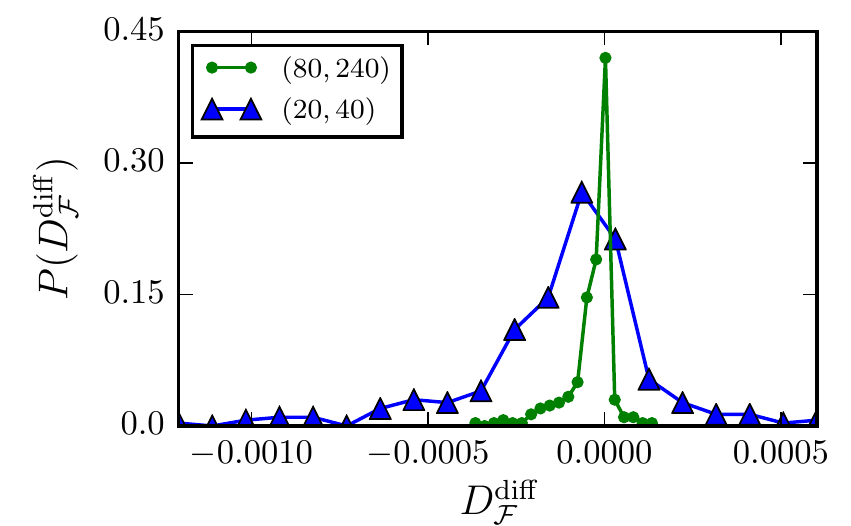}
  \includegraphics[width=0.95\columnwidth]{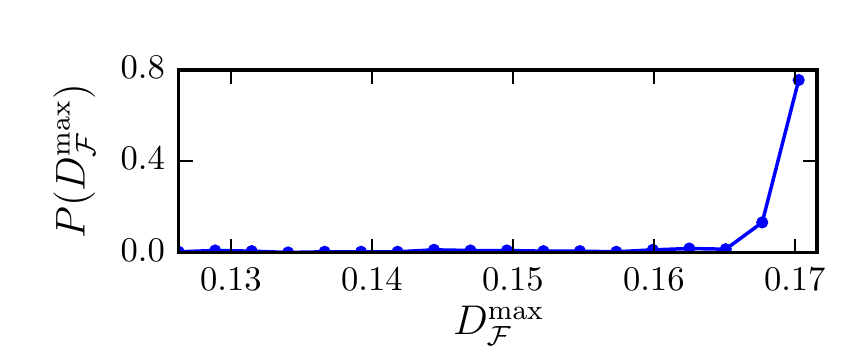}
  \caption{ (Top) Distribution $P(D_\mathcal{F})$ of the difference between the interaction distance of random state $\rho$ and its doubly degenerate realisation $\frac{1}{2}(\rho\oplus\rho)$.
The sample contains $300$ instances of $\rho$.
Increasing the number of basins used in the minimisations $D_\mathcal{F}(\rho)$
and $D_\mathcal{F}(\frac{1}{2}(\rho\oplus\rho))$ respectively, indicated as a tuple (legend), drives the distribution to a sharp peak.
(Bottom) Distribution of outputs of the maximisation of $D_\mathcal{F}$ performed $400$ times for spectra of rank randomly selected between $3$ and $8$.
Here, the number of basins in the minimisation of $D_\mathcal{F}$ is $30$.
The number of basins for the maximisation is randomly selected from the set $\{5,15,50\}$.
The maximisation achieves the conjectured upper bound $3-2\sqrt{2}\approx 0.1716$ but never exceeds it.
}
  \label{fig:Conjectures}
\end{figure}

{\sl Entanglement lattices.--} Here we present a view of the structure of free-fermion spectra in terms of order theory,
which aids in deriving exact results for $D_\mathcal{F}$ in special cases of flat spectra.
Each level of a free entanglement spectrum $\sigma$ is labelled by a set of occupation numbers.
We can define a partial order $a \preceq b$ for occupation-pattern labels $a$ and $b$ to hold if and only if $\sigma_a \ge \sigma_b$ for any free entanglement spectrum $\sigma$.
This partial order induces an equivalent covering relation $a \rightarrow c$ if and only if for any $\sigma$ we have $\sigma_a \ge \sigma_c$ and there exists no other pattern $b$ such that $\sigma_a \ge \sigma_b \ge \sigma_b$.
In Fig.~\ref{fig:f3young} the vertices are the occupation patterns and the arrows connecting them are the covering relations, forming an example of a Hasse diagram~\cite{priestley2002}.
The example in Fig.~\ref{fig:f3young} is the free entanglement partial order $\mathcal{F}_3$ for the case of three modes.

The partial ordering of the occupation-pattern coordinates is the strongest ordering compatible with any spectrum $\sigma$ which assigns to each mode $k$ an energy $E_k$ which is non-decreasing in $k$.
This ordering is known as the dominance order~\cite{Brylawski:1973bk}.
Another way of describing each level is to instead count the occupation of each energy gap.
Occupying an energy mode is equivalent to occupying all the energy gaps beneath it, and the transformation between the two pictures is known as majorisation.
Each energy-gap is non-negative but otherwise unconstrained thus the strongest partial order compatible with all $\sigma$ follows a coordinate order in these majorised coordinates.
Majorisation theory is a familiar topic in quantum information and has been the source of many fruitful applications, see Refs.~\cite{Nielsen:1999zza,Vidal:2000ci} for example.

Each set of majorised coordinates can be related to a Young diagram where the number of boxes in each column is the corresponding majorised coordinate and the number of boxes in each row labels the energy level occupied by each excitation (see Fig.~\ref{fig:f3young}).
Choosing a covering relation to mean the addition of a single box to the Young diagram defines a lattice known as Young's lattice,
a sublattice of which can be identified as the free entanglement lattice.
If we were considering free boson models instead of free fermions there would be no exclusion of multiple occupancy and the entanglement lattice would be isomorphic to Young's lattice.

\begin{theorem}
  \label{LatticeRGThm}
  \label{thm:m2_to_m1}
  $D_\mathcal{F}(\bar\rho(2^{n+1} - 2)) = D_\mathcal{F}(\bar\rho(2^n-1))$ for any integer $n \ge 1$
\end{theorem}
\begin{proof}
  The renormalisation step of the free entanglement lattice shown in Fig.~\ref{fig:f3young} is applied to $D_\mathcal{F}$ by using the triangle inequality on paired contributions to the cost function
  \begin{align}
    \sum_{a=0,1} \left|\rho_{\pi(a,b)} - \sigma_{(a,b)}\right|
    & \ge \left|\sum_{a=0,1}\rho_{\pi(a,b)} - \sum_{a=0,1}\sigma_{(a,b)}\right| \nonumber\\
    & \ge \left|\rho'_{\pi(b)} - \sigma'_{(b)}\right|,
  \end{align}
  where $\pi$ is a map from the occupation numbers of $\sigma$ to the eigenvalue ordering of $\rho$.
 We denote with $a$ the occupation number of a distinguished mode and $b$ is a tuple carrying the occupancies of all the other modes.
  When the degeneracy is $N=2^{n+1}-2$ there is no ambiguity in the assignment of occupations to the spectrum of $\rho$ because the final two levels are always the fully occupied pattern $(1,1,\dots)$ and the same but without a particle in the lowest energy mode $(0,1,\dots)$.
  By integrating over the occupations of the lowest energy mode in this way we produce a renormalised $\sigma'$ and the renormalised flat spectrum $\rho'$ which now has degeneracy $N'=2^n-1$.
  For the case of three fermionic modes this is illustrated in Fig.~\ref{fig:f3young} where the dashed ellipses group levels that are integrated together in this procedure.
  This provides a lower bound on $D_\mathcal{F}(\bar\rho(2^{n+1}-2))$ which is the same as the upper bound found by considering the direct sum ansatz $\frac{1}{2}(\sigma\oplus\sigma)$ demonstrating equality.
\end{proof}

\begin{figure}[th!]
  {\resizebox{\columnwidth}{!}{\includegraphics{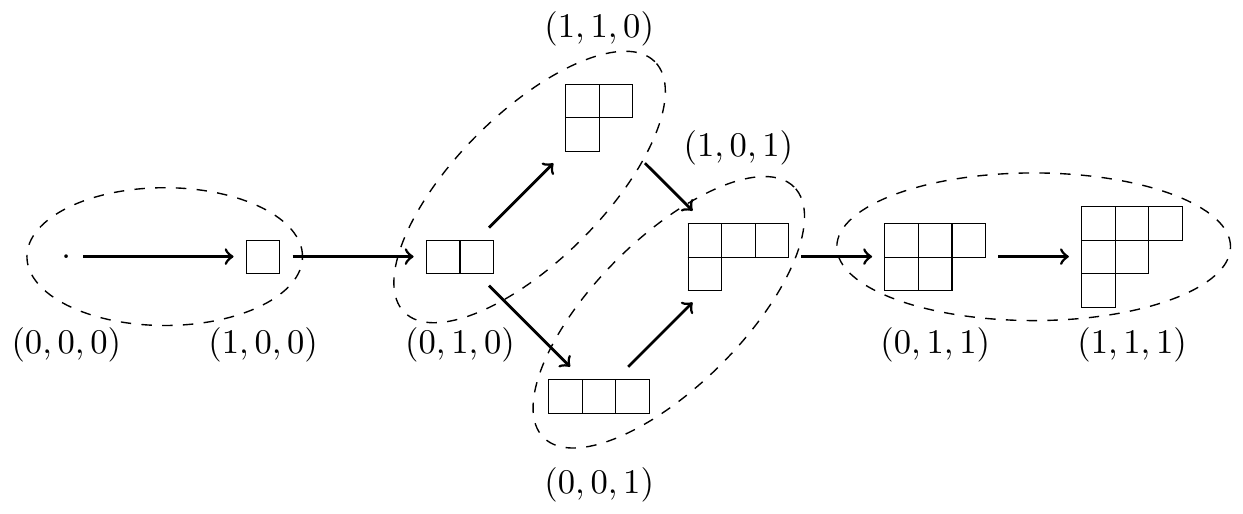}}}
  \caption{Hasse diagram for the free fermion entanglement lattice in the case of 3 fermionic modes ($\mathcal{F}_3$).
  Vertices contain the Young diagrams which represent the occupation numbers $(n_1,n_2,n_3)$ of the three modes with which they are labelled.
  Grouped vertices are renormalised by integrating over the first occupation number to give the vertices of $\mathcal{F}_2$.}
  \label{fig:f3young}
\end{figure}

Considering that we have an exact result $D_\mathcal{F}(\bar\rho_{(3)})=1/6$~\cite{Turner2017} we find as a corollary that $D_\mathcal{F}(\bar\rho_{(6)}) = D_\mathcal{F}(\bar\rho_{(3)}) = 1/6$ and hence the flat-spectrum upper-bound is exact for $N=6$.
This is the largest $N$ for which we have such a result besides powers of $2$.
  
We refer to this technique as entanglement lattice renormalisation and it appears a promising direction for analytical methods to derive lower-bounds to $D_\mathcal{F}$, something which is challenging since it is defined by a minimisation problem.
The difficulty in applying this method generally lies in dealing with all the different linear orderings compatible with the partial order.
\autoref{thm:m2_to_m1} falls in the special case where the ordering could be ignored, making it relatively simple.
    
{\sl Local equivalence between parafermion chain states.--}
We have found that a $\mathbb{Z}_{pq}$ chain and $\mathbb{Z}_p\times\mathbb{Z}_q$ chain have equivalent entanglement and that this leads to $D_\mathcal{F}=0$ for any eigenstate associated to any bipartition of a fixed-point $\mathbb{Z}_{2^n}$ chain.
We may wonder how this equivalence for all cuts could be coherently extended to an equivalence between states.
This can be achieved for the fixed-point parafermion chains as we will demonstrate.
In particular this means every eigenstate of a fixed-point $\mathbb{Z}_4$ chain is equivalent to a $\mathbb{Z}_2\times\mathbb{Z}_2$ chain according to a local unitary.

This problem can be approached by first constructing a matrix product state (MPS) representation
for the ground states of parafermion chains~\cite{ParaMPS}.
A translationally-invariant MPS representation is a decomposition,
\begin{equation}
  \ket{\psi} = \sum_{i_1,\cdots,i_L}\left(\eta\right|\Bigg(\prod_{j=1}^{L}\Gamma_{i_j}\Bigg)\left|\rho\right) \ket{i_1i_2\cdots{}i_L}
  \label{eqn:mps}
\end{equation}
into site-tensor $\Gamma$ and environment vectors $\left|\rho\right)$ and $\left(\eta\right|$~\cite{fannes1992,OrusPracticalMPSIntro}.
At the topological fixed point of a $\mathbb{Z}_N$ parafermion chain the site-tensor can be chosen as
\begin{align}
  \begin{tikzpicture}[baseline=-2ex]
    \node (left) at (-0.7,0) {$j$};
    \node (right) at (+0.7,0) {$k$};
    \node (gamma) [draw] at (0, 0) {$\Gamma$};
    \node (down) at (0, -0.7) {$i$};
    \draw (gamma) -- (down);
    \draw (left) -- (gamma);
    \draw (right) -- (gamma);
  \end{tikzpicture}
  &= \frac{1}{\sqrt{N}}\,\delta_{i+j-k\,(\mathrm{mod}\,{N})}
  \text{,}
\end{align}
where the basis indexed by $i$ is the basis of eigenvectors of $\tau$, with $\tau\ket{i}=\omega^i\ket{i}$.
The bases indexed by $j$ and $k$ are contracted and summed over as the matrix product in \rf{eqn:mps}.
The environment vectors can be chosen as $\left(\eta\right| = \left(0\right|$ and $\left|\rho\right)=\left|q\right)$ where $q$ is the parafermionic parity of $\ket{\psi}$.
It's simple to verify that the state \rf{eqn:mps} is indeed normalised and also in the ground subspace of all the terms in the Hamiltonian.

We turn the equation which states the local transformation between the two states into an equation involving a gauge transformation, $A$, internal to the MPS, yielding
\begin{equation}
  \begin{tikzpicture}[baseline=-2ex]
    \node (Al) [draw] at (-0.85,0) {$A$};
    \node (Ar) [draw] at (+1.0,0) {$A^{-1}$};
    \node (U) [draw] at (0, -0.7) {$U$};
    \node (left) at (-1.5,0) {};
    \node (right) at (+1.8,0) {};
    \node (gamma) [draw] at (0, 0) {$\Gamma_4$};
    \node (down) at (0, -1.3) {};
    \draw (gamma) -- (U);
    \draw (U) -- (down);
    \draw (Al) -- (gamma);
    \draw (Ar) -- (gamma);
    \draw (left) -- (Al);
    \draw (right) -- (Ar);
  \end{tikzpicture}
  =
  \begin{tikzpicture}[baseline=-2ex]
    \node (left) at (-1,0) {};
    \node (right) at (+1,0) {};
    \node (gamma) [draw] at (0, 0) {$\Gamma_{2\times2}$};
    \node (down) at (0, -0.7) {};
    \draw (gamma) -- (down);
    \draw (left) -- (gamma);
    \draw (right) -- (gamma);
  \end{tikzpicture}
\end{equation}
where $\Gamma_4$ and $\Gamma_{2\times2}=\Gamma_2\otimes\Gamma_2$ are site tensors for the $\mathbb{Z}_4$ and $\mathbb{Z}_2\times\mathbb{Z}_2$ chain respectively and $U$ is the local action of the unitary transformation between the states.
This is a potentially difficult problem to solve in general, however the states we consider are highly structured.

For each $i$ we get a matrix-slice of the site-tensor $\Gamma_i$ which transforms under the gauge transformations as $\Gamma_i \mapsto A\Gamma_iA^{-1}$.
For the states we consider, we can use this gauge transformation in order to choose a basis in which every matrix-slice of both $\Gamma_4$ and $\Gamma_{2\times2}$ is simultaneously diagonal.
In this way the site tensors can be interpreted as matrices, linear maps from the diagonal entanglement vector space spanned by $j=k$ to the physical vector space spanned by $i$.
The equation for local-unitary equivalence of $\mathbb{Z}_4$ and $\mathbb{Z}_2\times\mathbb{Z}_2$ parafermion chains then becomes
\begin{equation}
  U\,
  \frac{1}{2}
  \left(\begin{array}{cccc}
     1 &  1 &  1 &  1 \\
     1 & -i & -1 &  i \\ 
     1 & -1 &  1 & -1 \\
     1 &  i & -1 & -i
  \end{array}\right)
  =
  \frac{1}{2}
  \left(\begin{array}{cccc}
     1 &  1 &  1 &  1 \\
     1 & -1 &  1 & -1 \\ 
     1 &  1 & -1 & -1 \\
     1 & -1 & -1 &  1
  \end{array}\right)\text{.}
  \label{eqn:matofeigs}
\end{equation}
This clearly has a solution for $U$ as both matrices are unitary.
The same method applies to relating $\mathbb{Z}_{pq}$ and $\mathbb{Z}_p\times\mathbb{Z}_q$ chains for any $p$ and $q$.

We now turn to examine excited states.
The $\tau$ operator creates a pair of excitations, a twist-antitwist pair. \red{...}
\begin{align}
  \begin{tikzpicture}[baseline=-3ex]
    \node (left) at (-0.7,0) {$j$};
    \node (right) at (+0.7,0) {$k$};
    \node (gamma) [draw] at (0, 0) {$\Gamma$};
    \node (tau) [draw] at (0, -0.7) {$\tau$};
    \node (down) at (0, -1.3) {$i$};
    \draw (gamma) -- (tau);
    \draw (left) -- (gamma);
    \draw (right) -- (gamma);
    \draw (tau) -- (down);
  \end{tikzpicture}
  &= \frac{1}{\sqrt{N}}\,\omega^i\delta_{i+j-k\,(\mathrm{mod}\,{N})}
\end{align}
which is diagonal in the same gauge as the unexcited site-tensor.
The effect of $\tau$ on the matrices of eigenvalues for $\mathbb{Z}_{pq}$ is to pre-multiply with a diagonal matrix $\mathrm{diag}(1,\omega,\omega^2,\cdots)$.
Considering the composite system $\mathbb{Z}_2\times\mathbb{Z}_2$ chain we can create excitations on either of the two chains.
The diagonal matrix thus formed features only $1$ and $-1$  as opposed to the powers of $\omega = e^{2\pi{}i/4}$ found for the $\mathbb{Z}_4$ case.
If we demand that Eq.\rf{eqn:matofeigs} is satsified after inserting diagonal unitaries it must be the case that the two diagonal matrices have the same eigenvalues.
For the case of $\mathbb{Z}_4$ and $\mathbb{Z}_2\times\mathbb{Z}_2$ this clearly isn't possible as their eigenvalues are distinct.
This is a special case of $\mathbb{Z}_p\times\mathbb{Z}_q$ for $p$ and $q$ with shared prime factors.
In conclusion, a single unitary cannot be chosen in this case to map all the eigenstates of the $\mathbb{Z}_p\times\mathbb{Z}_q$ chain onto the $\mathbb{Z}_{pq}$ chain.

This is perhaps unsurprising because no symmetric finite depth unitary circuit can map between distinct symmetry-protected topological phases~\cite{WenSPT}.
However, our situation is slightly different as the unitary we consider need not be symmetry preserving.
Furthermore, its locality is 
particularly restricted as a tensor product of site-local unitaries.
Even in this setting there does not exist a single unitary which transforms all the eigenstates of the $\mathbb{Z}_4$
chain to those of $\mathbb{Z}_{2}\times\mathbb{Z}_{2}$ chain.
Nevertheless, it is possible to find a different unitary in general for any single eigenstate.
Similarly,
for the string-net models it is clear that a unitary connecting $\mathbb{Z}_4$ and $\mathbb{Z}_2\times\mathbb{Z}_2$ for all states cannot exist because these topologically ordered phases cannot be adiabatically connected~\cite{WenLocalUnitaries}, but this doesn't preclude the possibility of local unitaries for individual states.

{\sl Distribution functions for Abelian string-nets.--} Let the degeneracy of the flat entanglement spectrum be $\chi = N^k$ where $N$ is the order of the associated Abelian group and $k = |\partial{A}| - b_0$ is the number of independent degrees of freedom on the $b_0$-component boundary, $\partial A$.
We seek the distribution $P(D_\mathcal{F})$ for these spectra obtained by varying either $N$ or $k$.
As an intermediate step we will find the density function for $P(a)$, or equivalently $P(\alpha)$, where $a = \log_2 \chi \pmod{1}$ and $\alpha = 2^a$ are variables describing the position of $\chi$ between powers of $2$.

Consider first evaluating $D_\mathcal{F}$ for varying $k\in \mathbb{N}$, i.e. the subsystem size,
while keeping $N$ fixed.
If $\log_2 N$ is irrational then $a = k \log_2 N \pmod{1}$ for each $k$, which uniformly samples the interval $[0,1]$ and hence $P(a) = 1$ over that interval.
In Fig.~\ref{fig:p(a) for varying k} we compare this prediction for fixed $N$ against the sample of $k$ up to $2^{15}-1$.
\begin{figure}[!htb]
  \includegraphics[width=0.95\columnwidth]{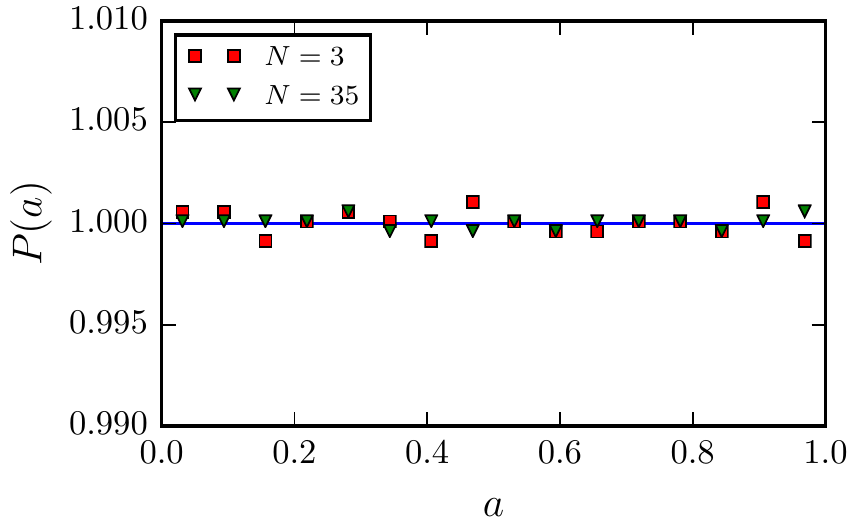}
  \caption{
    Distribution $P(a)$ for varying $k$ in the cases of $N=3$ and $N=35$.
    The line is the analytic result $P(a)=1$.
    The markers are numerical results for a sample up to $k=2^{15}-1$.
  }
  \label{fig:p(a) for varying k}
\end{figure}

If we instead fix $k=1$ and vary $N\in \mathbb{N}$,
while mapping each
interval between powers of $2$
back onto the interval $\alpha \in [1,2]$,
we find that it is densely and uniformly sampled in the limit of the infinite sequence, giving $P(\alpha)=1$.
The case $k=1$ refers to a rather unphysical subsystem containing zero vertices and one edge going through the cut, which is however equivalent to
cutting a single interval of lattice sites from the parafermion chain.
    
The situation is more complicated for $k > 1$.
For instance,
when $k=2$, region $A$ corresponds to
a disk enclosing a single trivalent vertex.
These sequences of spectra are also found in the parafermion chain when
region $A$ comprises $k$ disjoint intervals.
We first collapse all the natural numbers into the interval $[1,x]$ where $x = 2^{1/k}$ in the same manner as for $k=1$ to find $P(\beta)=1/(x-1)$, where $\beta$ identifies a corresponding real number for that integer in the interval.
Notice that $a = \log_x \beta$ and thus
\begin{equation}
  P(a) = P(\beta) \frac{\mathrm{d}\beta}{\mathrm{d}a} = \frac{\ln{2}}{k\left(2^{1/k}-1\right)} 2^{a/k}\text{.}
  \label{eqn:p(a) for varying n}
\end{equation}
In Fig.~\ref{fig:p(a) for varying n} we show the validity of our prediction for fixed $k$ for sample obtained by varying $N$ up to $2^{18}$.

\begin{figure}[!htb]
  \includegraphics[width=0.95\columnwidth]{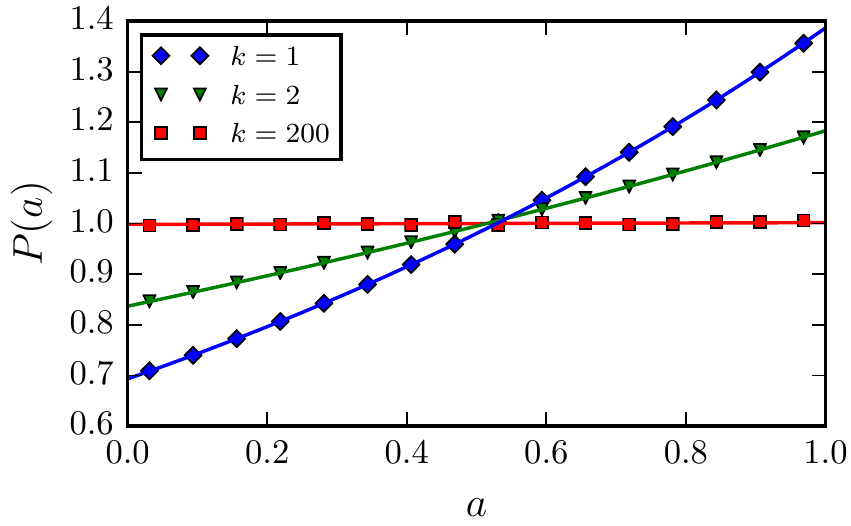}
  \caption{
    Distribution $P(a)$ obtained by varying $n$ for $k=1$,$2$, and $200$.
    The curves are the analytic result of \rf{eqn:p(a) for varying n}.
    The markers are numerical results for a sample up to $N=2^{18}$.
    Statistical errors are too small to be visible.
    We see that as $k$ is increased, $P(a)$ approaches the uniform distribution as expected.
  }
  \label{fig:p(a) for varying n}
\end{figure}

Armed with the distributions $P(a)$ or $P(\alpha)$ for the processes under discussion we can straightforwardly calculate the corresponding distributions for $D_\mathcal{F}$ for fixed $N$, 
\begin{equation}
P(D_{\mathcal{F}})=\frac{2/\ln 2}{\sqrt{1+D_{\mathcal{F}}(D_{\mathcal{F}}-6)}},
\end{equation}
which was presented in the main text, and for fixed large $k$ which is
\begin{equation}
  P(D_\mathcal{F}) = \frac{3 - D_\mathcal{F}}{\sqrt{1 + D_\mathcal{F}(D_\mathcal{F}-6)}}\text{.}
\end{equation}
Expanding the formula for fixed $k$ in a Taylor series for small $1/k$ we find
\begin{equation}
  P(a) = 1 + \left(a - \frac{1}{2}\right)\frac{\ln2}{k} + \frac{(a - \frac{1}{2})^2 - \frac{1}{12}}{2}\left(\frac{\ln2}{k}\right)^2 + O\left(\frac{1}{k^3}\right),
\end{equation}
and hence
\begin{equation}
  P(a) + P(1-a) = 1 + O\left(\frac{1}{k^2}\right)\text{.}
\end{equation}
This implies that corrections to the large-$k$ form of $P(D_\mathcal{F})$ are of order $1/k^2$ because $P(D_\mathcal{F})$ is proportional to $P(a) + P(1-a)$ for $a$ such that $D_\mathcal{F}(\bar\rho(2^a)) = D_\mathcal{F}$.
These are the only two values of $a$ to produce the same $D_\mathcal{F}$ and $D_\mathcal{F}$ viewed as a function of $a$ is invariant under the transformation $a \mapsto 1-a$.

{\sl Walker-Wang models.--} While the string-net models can be directly generalised to three spatial dimensions, a more powerful generalisation has been recently given in terms of the Walker-Wang models~\cite{WalkerWang}. These models allow non-trivial braiding of the charges giving a rich behaviour in their bulk and at their boundary~\cite{CurtSteve,KeyBurn}. 

The entanglement spectrum for topologically trivial cuts of a Walker-Wang model can be found in the same way as for string-nets given by $\rho_a= \left( \prod_{j\in a}d_{x_j} \right) /\mathcal{D}^{2(|\partial{A}|-1)}$ in the main text.
Nevertheless, partitions with non-trivial boundary topology reveal novel correlation properties~\cite{Alex}. 
To identify their effect on the interaction distance we take the region $A$ with a boundary topologically equivalent to a torus, as shown in Fig.~\ref{fourfigs}~(Left). Among the allowed configurations in the ground state is a braiding of loops with charges $x$ and $y$ supported in $A$ and $B$ respectively, connected by a string of charge $z$ piercing $\partial A$~\cite{Alex}. Thus, the probability spectrum should now encode information about the non-trivial braiding of the charges.
In Fig.~\ref{fourfigs}~(Right) we show $D_{\mathcal{F}}$ for toroidal cuts of non-Abelian SU(2)$_k$ Walker-Wang models as a function of the level $k\geq2$.
Compared to the spherical cut we see that the interaction distance depends not only on the geometry of the cut but also on the topology of $\partial A$.
Its non-zero value indicates the necessity of interactions
for the existence of non-Abelian topological order also in three dimensions.

\begin{figure}
	\subfloat{\includegraphics[width = .43\linewidth]{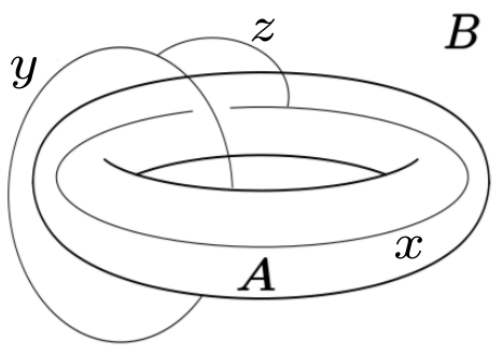}}
	\subfloat{\includegraphics[width = .55\linewidth]{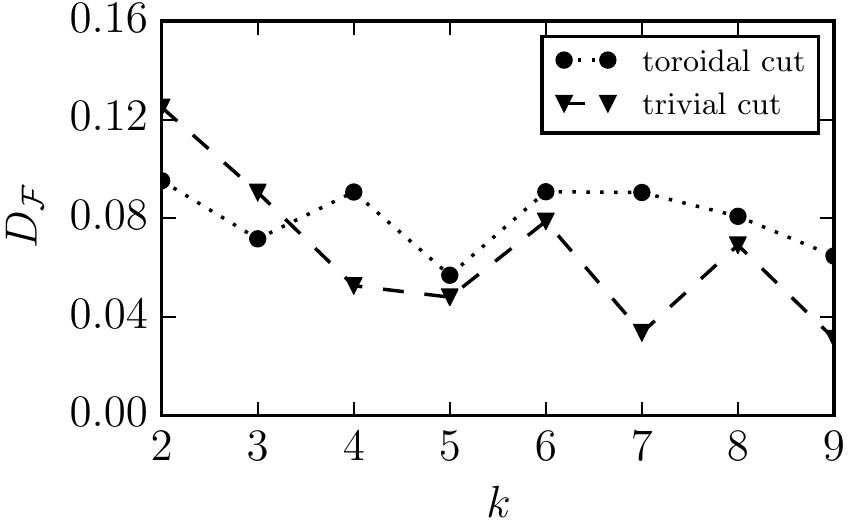}}
	\caption{Shows a topologically non-trivial toroidal cut of a Walker-Wang model (left, ~\cite{Alex}). A plot of $D_\mathcal{F}$ against $k$-level for a toroidal and a topologically trivial partition, both with $|\partial A|=3$ (right).}
	\label{fourfigs}
\end{figure}

\end{document}